\documentclass[11pt]{article}

\ifdefined\pdfmapfile
  \pdfmapfile{+cm.map}
  \pdfmapfile{+cmextra.map}
  \pdfmapfile{+symbols.map}
  \pdfmapfile{+rsfs.map}
  \pdfmapfile{+lm.map}
\fi
\usepackage{soul}

\usepackage{lmodern}
\usepackage[leqno]{amsmath}
\usepackage{amssymb}
\usepackage{amsthm}
\usepackage{enumerate}
\usepackage{verbatim}
\usepackage{mathrsfs}
\usepackage{dsfont}
\usepackage{graphicx}
\usepackage{fullpage}
\usepackage{color}

\newcommand{\mr}{\mathrm}

\let\H\undefined
\let\Ex\undefined
\let\Pr\undefined
\newcommand{\H}{\probabilityfont{H}}
\DeclareMathOperator*{\Ex}{\mathds{E}}
\DeclareMathOperator*{\Pr}{\mathds{P}}

\newcommand{\Aut}{\mr{Aut}}

\usepackage{mathtools}

\newcommand{\fieldfont}[1]{\mathbb{#1}}

\newcommand{\F}{\fieldfont{F}}

\newtheoremstyle{theorem-style}
  {}
  {}
  {\slshape}
  {}
  {\bf}
  {.}
  {.5em}
  {}

\newtheorem{thm}{Theorem}[section]

\newtheorem{main-la}[thm]{Main Lemma}

\theoremstyle{definition}
\newtheorem{df}[thm]{Definition}

\theoremstyle{plain}
\newtheorem{theorem}[thm]{Theorem}
\newtheorem{proposition}[thm]{Proposition}
\newtheorem{lemma}[thm]{Lemma}
\newtheorem{corollary}[thm]{Corollary}
\newtheorem{conjecture}[thm]{Conjecture}

\theoremstyle{definition}
\newtheorem{definition}[thm]{Definition}

\theoremstyle{remark}
\newtheorem{remark}[thm]{Remark}

\newcommand{\E}{\Ex}

\newcommand{\Odd}{\textsl{Odd}}

\newcommand{\dist}{\mathrm{dist}}
\newcommand{\ACzero}{\mr{AC}^0}

\def\[#1\]{\begin{align*}#1\end{align*}}

\usepackage[numbers,sort&compress]{natbib}

\definecolor{linkgreen}{RGB}{0,90,65}

\usepackage[
    bookmarks=false,
    colorlinks=true,
    linkcolor=linkgreen,
    citecolor=linkgreen,
    urlcolor=linkgreen
]{hyperref}
\usepackage[nameinlink,noabbrev]{cleveref}
\usepackage{authblk}

\newcommand{\1}{\mathbf 1}

\newcommand{\OR}{\operatorname{OR}}

\newcommand{\XOR}{\mathsf{XOR}}

\newcommand{\PARITY}{\mr{PARITY}}

\DeclareMathOperator{\girth}{girth}

\newcommand{\supp}{\operatorname{supp}}
\newcommand{\Part}{\mathsf{Part}}
\newcommand{\Chop}{\mathsf{Chop}}
\newcommand{\ELEMX}{\mathsf{ELEMX}}
\newcommand{\FtwoELEMX}{\ELEMX_{\F_2}}
\newcommand{\UR}{\mathsf{UR}}

\providecommand{\submissionmode}{0}

\title{On Top-Down and Local Lower Bounds for $\mathrm{AC^0}$ Circuits}

\ifnum\submissionmode=1
  \author{Anonymous authors}
  \date{}
\else
  \author[1,2]{%
    G\"ulce Karde\c{s}%
    \thanks{\href{mailto:gulcekardes@gmail.com}
    {{\texttt{gulcekardes@gmail.com}}}}%
  }
  \author[3]{%
    Benjamin Rossman%
    \thanks{\href{mailto:benjamin.rossman@duke.edu}
    {{\texttt{benjamin.rossman@duke.edu}}}}%
  }
  \affil[1]{University of Colorado Boulder}
  \affil[2]{Santa Fe Institute}
  \affil[3]{Duke University}
  \date{}
\fi

\begin{document}

\maketitle

\begin{abstract}
Classical lower bounds for $\mathrm{AC^0}$ circuits proceed bottom-up by simplifying or approximating gates beginning at the input layer.  We introduce a complementary top-down model called the \emph{Chopping Game}, played by adversaries Spoiler and Duplicator on the sets of $0$- and $1$-inputs of a Boolean function.  In each round, Spoiler keeps at least a $1/m$-fraction of one side, and Duplicator arbitrarily restricts the other; Spoiler seeks to minimize (and Duplicator to maximize) the number of rounds until some coordinate separates the two remaining sets.  Every depth-$d$, fan-in-$m$ circuit induces a $d$-round winning strategy for Spoiler, while Duplicator strategies that survive $d$ rounds formalize top-down lower-bound arguments.  

Through the Chopping Game and using the polynomial-approximation method, we first obtain the classical lower bound for depth-$d$ $\mathrm{AC^0}$ circuits in a top-down fashion.  We then consider a \emph{$k$-local} variant of the Chopping Game, which relaxes Spoiler's win condition by requiring a separating coordinate within each Hamming ball of radius $k$, rather than a single coordinate globally.  We put forward a conjecture that the $d$-round $k$-local Chopping Game for $\mathrm{PARITY}$ requires $m = n^{\omega(1)}$ in the regime $d \ll k \ll n$.  We prove such a lower bound $m \ge n^{\Omega(k^{1/d}/d)}$ when Spoiler is restricted to so-called \emph{affine} strategies, a class of strategies that achieves the best known upper bounds.  Finally, we formulate a version of the $k$-local Chopping Game on $n$-regular graphs of girth $>2k$, and we conjecture a graph-theoretic analogue of ``$\mathrm{PARITY} \notin \mathrm{AC^0}$''.
\end{abstract}

\newpage

\section{Introduction}
\label{sec:introduction}

Essentially all lower bounds for $\ACzero$ circuits (with a few exceptions
for depth at most four, discussed below) are based on two methods introduced
in the 1980s: random restrictions and switching lemmas
\cite{Ajtai1983,FurstSaxeSipser1984,Yao1985,Hastad1986} and low-degree
polynomial approximation \cite{Razborov1987,Smolensky1987}.  Both methods
proceed in a {\em bottom-up} fashion by simplifying or approximating gates beginning at the
input layer.  Notable later developments based on these techniques include the
Linial--Mansour--Nisan theorem on the Fourier concentration of $\ACzero$
functions \cite{LinialMansourNisan1993} and Braverman's theorem that
polylogarithmically independent distributions fool polynomial-size
$\ACzero$ circuits \cite{Braverman2010}.  Nevertheless, 
completely new techniques are needed to make progress on central open problems such as proving
$\textsc{Inner Product}\notin\ACzero\circ\XOR$, or breaking the
$2^{\Omega(n^{1/(d-1)})}$ size lower-bound barrier for depth-$d$ circuits.

One prominent line of work aims to prove novel {\em top-down} lower bounds through the Karchmer--Wigderson (KW) Game characterization of $\ACzero$ circuits \cite{KarchmerWigderson1990}.  In the KW Game for a Boolean function $f$,
Alice and Bob receive a one-input $x\in f^{-1}(1)$ and a zero-input
$y\in f^{-1}(0)$ and must find a coordinate $i$ such that $x_i\ne y_i$.
Depth-$d$, fan-in-$m$ circuits for $f$ are precisely characterized by
$d$-round protocols in which each message takes at most $m$ possible
values.  Such a protocol describes a descent from the output gate of a
circuit toward an input literal, with each message specifying a child that certifies the current value.

A natural top-down approach is to follow, in each round, a message shared by
the largest number of active inputs (i.e.,\ the least informative message).  This
approach has proved successful for monotone circuits, achieving optimal lower bounds on the depth of fan-in $2$ monotone circuits for $st$-Connectivity \cite{KarchmerWigderson1990} and Clique \cite{GoldmannHastad1992}.
For non-monotone $\ACzero$ circuits,
exponential lower bounds by top-down methods were obtained for circuits of depth $3$ \cite{HastadJuknaPudlak1995,MW19} and depth $4$ \cite{GoosRiazanovSofronovaSokolov2023}.

\subsection{The Chopping Game}

We introduce the Chopping Game as a framework 
to formalize this kind of top-down argument.
The initial position in the game is the pair of sets $(A_0,B_0)$ where $A_0$ and $B_0$ are the sets of one-inputs and zero-inputs of a Boolean function $f$.
In an odd round, Spoiler retains at
least a $1/m$ fraction of the surviving $A$-side and Duplicator arbitrarily
prunes the $B$-side; in an even round, Spoiler chops the $B$-side and
Duplicator prunes the $A$-side.  Spoiler wins when one coordinate separates
the two remaining sets.

At a gate in a KW descent, the next message partitions the active side
into at most $m$ parts. Spoiler follows a largest part, while Duplicator
is free to prune the opposite side.  The Chopping Game relaxes this step,
by allowing Spoiler to choose any subset of the same density, without
specifying the other parts of the partition.  It follows that every
depth-$d$, fan-in-$m$ circuit (with output gate OR) induces a $d$-round winning strategy for
Spoiler in the orientation determined by its output gate. Thus, if
Duplicator survives $d$ rounds, then no such circuit
computes $f$. 

The converse fails: Spoiler's freedom to choose a dense set can give a
Chopping strategy which does not arise from any circuit descent (see the
example following Proposition~\ref{prop:partition-chopping}).  A lower bound
for the Chopping Game is therefore stronger than the corresponding circuit
lower bound.

\subsection{Our results and conjectures}

We first use the Chopping Game with the polynomial-method to obtain a lower bound 
in top-down form.  
Our focus then shifts to local versions of the KW and Chopping Games on
the Boolean cube and on regular graphs.

\paragraph{Top-down approximation.}

We use the Chopping Game to give a proof of $\PARITY\notin\ACzero$ through a top-down use of polynomial approximations. For fixed $d\ge2$, 
we directly
construct 
from a $d$-round winning Spoiler strategy for parity 
a probabilistic
polynomial over $\F_3$ of degree $O((\log m)^{d-1})$ that distinguishes
even- and odd-weight strings with constant advantage. The
$\Omega(\sqrt n)$ approximation-degree lower bound for parity over $\F_3$
\cite{Smolensky1987} then 
gives
$m=2^{\Omega(n^{1/(2(d-1))})}$. 
That is, we give a top-down implementation of the classical approximation-method lower bound for parity.

\paragraph{The Distance-$k$ KW Game.}

In the Distance-$k$ KW Game, Alice and Bob receive strings
$x,y\in\{0,1\}^n$ at Hamming distance exactly $k$ and, following a deterministic communication protocol, must find a
coordinate on which they differ.  Note that here the promise on $(x,y)$ is not a
combinatorial rectangle.
When $k$ is odd (and assuming $|x|$ is odd), this is a non-rectangular subproblem of the KW relation for parity.
The Distance-$k$ KW Game also makes sense when $k$ is even, although it is more closely related to the Universal Relation (as we discuss below).

For every $k$, there is an elementary two-round protocol using
at most $(k+1)\lceil\log n\rceil$ bits in total. Alice sends the syndrome of
$x$ with respect to a parity-check matrix of a binary code of distance at
least $2k+1$ \cite{BoseRayChaudhuri60}; using $y$, Bob decodes $x+y$ and
returns any one of its nonzero coordinates. (Note that this upper bound does {\em not} correspond to any depth-$2$ circuit, or indeed any circuit at all.) We prove a
$(k+1)\log n-O_k(1)$ lower bound on the total cost of every two-round protocol.
For every fixed even $k$, we also prove a round-independent lower bound
$\frac{k}{2}\log n-O_k(1)$ on total communication.

We then advance a conjecture about the constant-round complexity of the Distance-$k$ KW Game for fixed $k$, which we believe isolates a lower-bound question that calls for new techniques.

\begin{conjecture}
\label{conj:distance-k-informal}
For all $d \ge 2$ and $d \ll k \ll n$,
every $d$-round $m$-ary protocol for the Distance-$k$
KW Game requires $m=n^{\omega(1)}$.
\end{conjecture}

So far as we know, a lower bound $m = n^{\Omega(k/d)}$ could even be possible (moreover, even for the corresponding Distance-$k$ Chopping Game). Here we modestly conjecture a lower bound on $m$ that is super-polynomial in $n$ for any fixed $d$ as $k$ grows; the quantifiers are stated explicitly in
Conjecture~\ref{conj:distance-k}. Even in this quantitatively weak form, in the case of odd $k$, Conjecture \ref{conj:distance-k-informal} is stronger than
$\PARITY\notin\ACzero$: it asks for a super-polynomial lower bound 
a single fixed-distance
restriction of the KW relation for PARITY.
Nevertheless, the classical $\ACzero$ lower bound techniques do not appear to apply to this question.

In the Distance-$k$ Chopping Game, Spoiler wins when every point on either
surviving side has a coordinate which separates it from all surviving points
on the opposite side at distance exactly $k$.  We also consider a $k$-local
version in which this condition is required for all opposite-side points
within Hamming distance $k$.  We make the same qualitative conjecture
$m=n^{\omega(1)}$ for the more top-down Distance-$k$ Chopping Game.

\paragraph{Linear KW protocols and affine chopping strategies.}

We call a KW protocol {\em linear} if every message consists of $\F_2$-linear
functions of the speaker's input, and we call a Chopping Game strategy {\em affine} if every
chop is an intersection with an affine subset of $\F_2^n$.  
The two-round
upper bound above has a linear variant: Alice and Bob each send the syndrome
of their input, after which both recover $x+y$.  For $d \ge 2$, 
a lower bound in 
\cite{Rossman2017Subspace} shows that every $d$-round $\F_2$-linear protocol
for parity has a message of length at least $n^{1/(d-1)}-1$ bits.  The
standard $d$-round top-down strategy for parity likewise uses affine
chops, with $\log m=O(n^{1/(d-1)})$.

Our main technical result shows that every $d$-round affine
winning strategy in the $k$-local Chopping
Game for parity requires
$m\ge n^{\Omega(k^{1/d}/d)}$.  The best known upper bounds use affine
strategies.  For every $d\ge2$ and infinitely many $k$, we also prove that
every $d$-round $\F_2$-linear protocol for the
Distance-$k$ KW Game has a
message of length at least
$\frac{k^{1/(d-1)}}{d-1}\log n-O_k(1)$ bits.  The proof uses a fixed-core
sunflower bound in a round-elimination argument.

\paragraph{Distance-$k$ Conjecture in high-girth graphs.}

Let $G$ be a finite $n$-regular graph of girth greater than $2k$, so that
vertices at distance $k$ have a unique shortest path.  In the
graph-theoretic Distance-$k$ KW Game, Alice and Bob
receive vertices $x,y$ at distance $k$, and each identifies the neighbor of
their own vertex on the unique shortest $x$-to-$y$ path.  

We formulate a version of the Distance-$k$
Conjecture on such graphs $G$
(Conjecture~\ref{conj:high-girth-distance}), which we view as a
graph-theoretic version of $\PARITY\notin\ACzero$.  An elementary coloring
argument gives a matching two-round upper bound of
$2k\log n+O(1)$ bits per message.  Finiteness matters: after choosing a
public root and local port labels, the corresponding relation on the infinite
$n$-regular tree has a three-round protocol of total cost
$\log n+O(\log k)$. While is no formal reduction between the graph and Hamming cube
conjectures, we believe this new graph-theoretic question cleanly isolates the challenge of proving complexity lower bounds from local principles. Optimistically, a proof of this conjecture could lead to new techniques in the circuit setting.

\subsection{Related work}

The local viewpoint in this paper is closest to the finite-limit method.
The idea goes back to Sipser's infinitary analogy and was developed into a
top-down lower-bound method for depth-three circuits by H{\aa}stad, Jukna,
and Pudl{\'a}k, and into a general finite-limit criterion for monotone
computations by Jukna
\cite{Sipser1984,HastadJuknaPudlak1995,Jukna1997FiniteLimits}.  Meir and
Wigderson later gave an information-theoretic version of the depth-three
argument, and G{\"o}{\"o}s, Riazanov, Sofronova, and Sokolov extended
top-down lower bounds to depth four
\cite{MW19,GoosRiazanovSofronovaSokolov2023}.  A key distinction, however, is the scale of locality. In these works, the relevant locality parameter grows with $n$: the depth-three arguments test local limits on
about $\sqrt n$ coordinates, while the depth-four argument uses Hamming
spheres of radius $n^{1/3}$ and block flips of size $n^{1/3+o(1)}$. Consequently, these methods do not yield nontrivial bounds for the fixed-$k$ Chopping Games studied here.

The Universal Relation asks for a differing coordinate under the sole
promise $x\ne y$.  Tardos and Zwick proved essentially tight
$\Theta(n)$-bit deterministic communication bounds for this problem
\cite{TardosZwick1997}.  The diagonal-transcript argument behind
Theorem~\ref{thm:tarui-even-k} is a fixed-distance analogue of their
lower-bound argument.  By contrast, the $2\lceil\log n\rceil$ protocol used at odd
distance applies only to the opposite-parity restriction of the Universal
Relation.

Certificate Games study zero-communication protocols for KW relations under
several models of shared correlation
\cite{ChakrabortyGalLaplanteMittalSunny2023}.  Our terminal positions instead
require deterministic output maps that are correct on every surviving promise
pair.  We use no quantitative comparison with the Certificate Game
parameters.

\subsection*{Outline of the paper}

Section~\ref{sec:preliminaries} defines the circuit and game models.
Section~\ref{sec:top-down-approximation} gives the polynomial and formula
approximation theorems.  Section~\ref{sec:distance-k} introduces the
Distance-$k$ and $k$-local games, states their main conjectures, and gives
the two-round and even-distance bounds.  Section
\ref{sec:linear-distance-k} proves the lower bounds for linear protocols and
affine Spoiler strategies.  Section~\ref{sec:graph-analogue} treats the
graph-theoretic conjecture, the coloring upper bound, and the infinite-tree
protocol.  Section~\ref{sec:future-directions} states the main open
questions.

\section{Preliminaries}\label{sec:preliminaries}

Throughout, $d\ge0$, $m\ge2$, and $n\ge1$ are integers.  We write
$[n]=\{1,\ldots,n\}$ and take all logarithms to base two.  For
$x\in\{0,1\}^n$, let $|x|$ and $\supp(x)$ denote its Hamming weight and
support; the Hamming distance between $x$ and $y$ is denoted $\dist(x,y)$.
Unless otherwise stated, $d$ and $k$ are fixed while $n$ tends to infinity.
The notation $d\ll k\ll n$ means that $d$ is fixed, $k$ is sufficiently
large in terms of $d$, and $n$ is sufficiently large in terms of $d,k$.
Constants implicit in asymptotic notation may depend on the parameters in
the subscript.

\subsection{Circuits and Karchmer--Wigderson Games}

\begin{definition}[$\ACzero$ circuits]
\label{def:ac0-circuits}
An $\ACzero$ circuit has input gates labeled by constants or literals and
internal gates labeled AND or OR.  Its depth is the maximum number of
internal gates on an input--output path, and its fan-in is the maximum
number of inputs to a gate.  A formula is a circuit in which every nonoutput
gate has fan-out one.  We take circuits and formulas to be alternating.  A
$\Sigma_d$ circuit has depth at most $d$ and an OR output gate; $\Pi_d$
circuits are defined dually.  Literals and constants have depth zero and
belong to both orientations.
\end{definition}

A depth-$d$, fan-in-$m$ circuit can be unfolded into a formula without
increasing either parameter; the resulting formula has at most $m^d$
leaves.  We therefore use fan-in for both circuits and formulas.

For $A,B\subseteq\{0,1\}^n$, a Boolean function \emph{separates} $A$ from
$B$ if it is one on $A$ and zero on $B$.  We regard $(A,B)$ as an ordered
pair, with $A$ as the one-side.  The position is \emph{globally terminal}
if a literal or constant separates $A$ from $B$.  In particular, a position
with an empty side is terminal.

Let $f:\{0,1\}^n\to\{0,1\}$, and put $A=f^{-1}(1)$ and $B=f^{-1}(0)$.
All communication protocols in this paper are deterministic.

\begin{definition}[KW Communication Game]
\label{def:kw-communication}
Alice receives $x\in A$ and Bob receives $y\in B$.  In a $d$-round
protocol they exchange at most $d$ alternating messages, beginning with
Alice.  At the end, both parties output the same coordinate $i\in[n]$ such
that $x_i\ne y_i$.  The protocol is $m$-ary if every message belongs to an
alphabet of size at most $m$, and it is $L$-bit if every message has length
at most $L$.  Its total cost is the maximum number of bits in a transcript.
Thus an $L$-bit, $d$-round protocol has total cost at most $dL$.
\end{definition}

An $m$-ary protocol may be encoded using $\lceil\log m\rceil$ bits per
round; conversely, an $L$-bit message has at most $2^L$ possible values.
When message lengths vary, we require the possible messages at each
transcript to be prefix-free.  Equivalently, the protocol may be expanded
into a binary tree whose maximum root-to-leaf length is its total cost.

\begin{theorem}[Karchmer--Wigderson \cite{KarchmerWigderson1990}]
\label{thm:kw-correspondence}
There is an Alice-first $d$-round $m$-ary protocol for the KW game of $f$
if and only if a $\Sigma_d$ formula of fan-in at most $m$ computes $f$.
The Bob-first version corresponds to $\Pi_d$ formulas.  The same statement
holds with circuits in place of formulas, by unfolding.
\end{theorem}

\subsection{The Partition and Chopping Games}

We next give the KW protocol tree as an adversarial partition game. A position is an
ordered pair $(A,B)$, and the first set is the side acted on in the next
round. The order reverses after every move.

\begin{definition}[Partition Game]
\label{def:partition-game}
Write $\Part_{d,m}(A,B)$ if Spoiler wins the following game in at most $d$
rounds.  A globally terminal position is an immediate win, while a
nonterminal position is a loss when $d=0$.  Otherwise Spoiler partitions
$A=A_1\sqcup\cdots\sqcup A_t$ into at most $m$ nonempty parts, Duplicator
chooses $j\in[t]$, and play continues from $(B,A_j)$.  Equivalently,
\[
 \Part_{d,m}(A,B)
 \quad\Longleftrightarrow\quad
 \exists\ A=A_1\sqcup\cdots\sqcup A_t\ (t\le m)\ \ \forall j\in[t]\quad
 \Part_{d-1,m}(B,A_j)
\]
at every nonterminal position.
\end{definition}

Here Spoiler fixes the message partition at the current node, and Duplicator
chooses the message that is sent.

\begin{proposition}[Partition Game and formulas]
\label{prop:partition-formula}
For all $A,B\subseteq\{0,1\}^n$, $\Part_{d,m}(A,B)$ holds if and only if a
$\Sigma_d$ formula of fan-in at most $m$ separates $A$ from $B$.
\end{proposition}

\begin{proof}
The assertion is immediate at depth zero.  Suppose that $d\ge1$ and
$(A,B)$ is nonterminal.  If Spoiler partitions
$A=A_1\sqcup\cdots\sqcup A_t$, then by induction there is a
$\Sigma_{d-1}$ formula $G_j$ separating $B$ from $A_j$ for every $j$.
The formula $\bigvee_{j=1}^t\neg G_j$ separates $A$ from $B$.

Conversely, let $F=\bigvee_{j=1}^t F_j$ be a $\Sigma_d$ formula separating
$A$ from $B$.  Assign each $x\in A$ to an index $j$ for which $F_j(x)=1$.
This partitions $A$ into at most $t\le m$ nonempty parts $A_j$.  Since
$F_j$ is zero on $B$ and one on $A_j$, the formula $\neg F_j$ separates
$B$ from $A_j$.  Induction gives Spoiler's continuation from $(B,A_j)$.
\end{proof}

We now introduce the Chopping Game, obtained by discarding most of the information in a partition move, in order to derive stronger circuit lower bounds.

\begin{definition}[Chopping Game]
\label{def:chopping-game}
Write $\Chop_{d,m}(A,B)$ if Spoiler wins the following game in at most $d$
rounds.  A globally terminal position is an immediate win, while a
nonterminal position is a loss when $d=0$.  Otherwise Spoiler chooses
$A'\subseteq A$ with $|A'|\ge \frac{|A|}{m}$; Duplicator chooses an arbitrary
$B'\subseteq B$; and play continues from $(B',A')$.  Equivalently,
\[
 \Chop_{d,m}(A,B)
 \quad\Longleftrightarrow\quad
 \exists\ A'\subseteq A,\ |A'|\ge \frac{|A|}{m}\ \ \forall B'\subseteq B\quad
 \Chop_{d-1,m}(B',A')
\]
at every nonterminal position.
\end{definition}

Thus Spoiler fixes $A'$ before seeing Duplicator's response.  The same chop
must work for every $B'\subseteq B$, although the continuation from
$(B',A')$ may depend on $B'$.  Several counterfactual responses to one chop
therefore describe different branches of the strategy, not successive moves
in one play.

\begin{proposition}[Partition implies chopping]
\label{prop:partition-chopping}
If $\Part_{d,m}(A,B)$ holds, then so does $\Chop_{d,m}(A,B)$.  Consequently, if
Duplicator survives the $d$-round $m$-Chopping Game from $(A,B)$, then no
$\Sigma_d$ circuit of fan-in at most $m$ separates $A$ from $B$.  If
Duplicator survives from both $(A,B)$ and $(B,A)$, then no alternating
depth-$d$, fan-in-$m$ circuit separates the two sets in either orientation.
\end{proposition}

\begin{proof}
We argue by induction on $d$.  Suppose that $(A,B)$ is nonterminal and
that Spoiler has a winning Partition Game strategy beginning with
$A=A_1\sqcup\cdots\sqcup A_t$, where $t\le m$.  Some part $A_j$ has size
at least $\frac{|A|}{m}$; Spoiler uses this part as the chop.  The Partition Game
strategy wins from $(B,A_j)$.  It therefore also wins from $(B',A_j)$ for
every $B'\subseteq B$, since a separating formula remains a separator after
restricting either side.  By induction, Spoiler wins the Chopping Game from
every such $(B',A_j)$.

The circuit consequences follow from
Proposition~\ref{prop:partition-formula}.  A $\Pi_d$ separator of $(A,B)$
is the complement of a $\Sigma_d$ separator of $(B,A)$.
\end{proof}

Unlike the Partition Game, the Chopping Game does not characterize circuits:
the implication in Proposition~\ref{prop:partition-chopping} cannot be
reversed.  For instance, let $B=\{0^n\}$ and
$A=\{0,1\}^n\setminus\{0^n\}$, where $n\ge3$.  Spoiler wins the one-round
$2$-Chopping Game by choosing $A'=\{x:x_1=1\}$, which has relative density
greater than one half.  For every $B'\subseteq B$, the reversed position
$(B',A')$ is terminal, separated by $\neg x_1$ when $B'$ is nonempty.  On
the other hand, a one-round two-part Partition strategy would require two
coordinates which meet the support of every nonzero string.  The unit vectors
show that no two coordinates suffice.  Thus Chopping strategies can be
strictly more general than circuit descents.

\section{Top-Down Approximation Theorems}
\label{sec:top-down-approximation}

For nonempty $S\subseteq\{0,1\}^n$, let $\Pr_S$ denote probability under the
uniform distribution on $S$.

\begin{df}[Proper polynomial]
A polynomial $p\in\F_3[x_1,\ldots,x_n]$ is \emph{proper} if
$p(x)\in\{0,1\}$ for every $x\in\{0,1\}^n$, where the Boolean values are
identified with their images in $\F_3$.
\end{df}

In this section we show that a winning Spoiler strategy on $(A,B)$ yields
two weak separators of $A$ from $B$: a proper polynomial and a shallow
formula. The first proof applies the Razborov--Smolensky polynomial
construction at each round of the Chopping Game, giving a top-down 
implementation
of the polynomial method. Replacing the algebraic OR in this argument by
an ordinary Boolean OR gives the formula statement.

\subsection{Polynomial approximation}

We work over $\F_3$ because parity is linear over $\F_2$ but has large
approximation degree over $\F_3$ \cite{Razborov1987,Smolensky1987}.  This
choice also makes the algebraic OR especially simple: for $a\in\F_3$,
$a^2$ is the indicator of $a\ne0$.

\begin{theorem}[Top-down polynomial approximation]
\label{thm:top-down-polynomial}
There is an absolute constant $C$ with the following property.  Let
$d\ge1$, $m\ge2$, and let $A,B\subseteq\{0,1\}^n$ be nonempty.  If
$\Chop_{d,m}(A,B)$, then there is a proper polynomial $p$ over $\F_3$ of
degree at most $\bigl(C\log m\bigr)^{d-1}$ such that
\[
 \Pr_A[p=1]\ge\frac{1}{2m}
 \qquad\text{and}\qquad
 \Pr_B[p=1]\le m^{-10}.
\]
\end{theorem}

The polynomial in Theorem~\ref{thm:top-down-polynomial} is a one-sided weak
separator: on an arbitrary pair $(A,B)$ it has a small acceptance probability
on $B$ and a positive, but possibly small, acceptance probability on $A$.
For parity, translation symmetry amplifies this gap to constant advantage.

\begin{proof}
Put $\rho=\frac{1}{2m}$ and $\beta=m^{-10}$, and let
$r=\lceil2m\ln(\frac{2}{\beta})\rceil$ and
$s=\lceil\log_3(\frac{2}{\beta})\rceil$.  Thus
$(1-\rho)^r\le\frac{\beta}{2}$, $r\beta\le\frac{1}{2}$, and
$3^{-s}\le\frac{\beta}{2}$.  We prove the theorem by induction on $d$, with
the more precise degree bound $(2s)^{d-1}$.

If $(A,B)$ is terminal, its separating literal is an exact proper
polynomial.  Suppose next that $d=1$ and the position is nonterminal.  Let
$A'$ be Spoiler's first chop.  Taking $B$ itself as Duplicator's response
leaves the terminal position $(B,A')$.  The complement of its separating
literal is one on $A'$ and zero on $B$, and
$|A'|\ge\frac{|A|}{m}$.

Now assume $d\ge2$, and again let $A'$ be Spoiler's first chop.  Set
$R_0=B$.  For $j=1,\ldots,r$, stopping if the residual set is empty, use
$R_{j-1}$ as Duplicator's response to this same chop.  The strategy wins
from $(R_{j-1},A')$ in at most $d-1$ further rounds.  Induction therefore
gives a proper polynomial $g_j$ of degree at most $(2s)^{d-2}$ such that
$\Pr_{R_{j-1}}[g_j=1]\ge\rho$ and
$\Pr_{A'}[g_j=1]\le\beta$.  Define
$R_j=R_{j-1}\cap\{g_j=0\}$.  Let $t\le r$ be the number of polynomials
constructed.  Either the residual is empty or $t=r$ and
$\frac{|R_t|}{|B|}\le(1-\rho)^r\le\frac{\beta}{2}$.

Here the responses $R_0,R_1,\ldots$ describe counterfactual branches below
the same first chop $A'$.  They are not Duplicator's successive moves in a
single play.

For $M=(c_{uj})\in\F_3^{s\times t}$, set
\[
 p_M=\prod_{u=1}^s
 \left[1-\left(\sum_{j=1}^t c_{uj}g_j\right)^2\right].
\]
On the Boolean cube this polynomial is Boolean-valued.  It is one whenever
all the $g_j$ vanish.  If $(g_1(x),\ldots,g_t(x))\ne0$ and $M$ is uniform,
then $p_M(x)=1$ with probability $3^{-s}$.  It follows that
$\E_M\Pr_B[p_M=1]\le\beta$, so fix a matrix $M$ attaining this
bound.  On the other hand, the union bound gives
$\Pr_{A'}[p_M=0]\le t\beta\le\frac{1}{2}$, independently of $M$.  Hence
$\Pr_A[p_M=1]\ge\frac{|A'|}{2|A|}\ge\frac{1}{2m}$.  Finally,
$\deg p_M\le2s(2s)^{d-2}=(2s)^{d-1}$.  Since
$s=O(\log m)$, the theorem follows.
\end{proof}

Let $O_n$ and $E_n$ denote the odd- and even-weight strings in $\{0,1\}^n$.

\begin{corollary}[Parity]
\label{cor:top-down-parity}
For every fixed $d\ge2$, if $\Chop_{d,m}(O_n,E_n)$, then
\[
 m=2^{\Omega(n^{1/(2(d-1))})}.
\]
The same lower bound therefore holds for the fan-in of depth-$d$
$\ACzero$ circuits computing parity.
\end{corollary}

\begin{proof}
Let $D=(C\log m)^{d-1}$, and take the polynomial $p$ supplied by
Theorem~\ref{thm:top-down-polynomial}.  For an even-weight string $z$, the
translate $p_z(x)=p(x\oplus z)$ has the same degree and is proper.  For
fixed $x\in O_n$ (respectively, $x\in E_n$), a uniform $z\in E_n$ makes
$x\oplus z$ uniform in $O_n$ (respectively, in $E_n$).  Choose independent
uniform $z_1,\ldots,z_{4m}\in E_n$ and independent uniform coefficients
$c_{uj}\in\F_3$, where $u\in\{1,2\}$ and $j\in[4m]$, and form the
algebraic OR
\[
 H=1-\prod_{u=1}^2
 \left[1-\left(\sum_{j=1}^{4m}c_{uj}p_{z_j}\right)^2\right].
\]
For every odd-weight input the false-negative probability is at most
$e^{-2}+\frac{1}{9}<\frac{1}{4}$, while for every even-weight input the
false-positive probability is at most
$4m\cdot m^{-10}<\frac{1}{4}$.  By averaging, one may fix the translates
and the two coefficient rows so that the resulting proper polynomial has
degree at most $4D$ and agrees with parity on more than three quarters of
the cube.  The standard
$\Omega(\sqrt n)$ lower bound for an $\F_3$ polynomial agreeing with parity
on at least three quarters of the cube
\cite{Smolensky1987} gives the result.  The same argument applies with
$O_n$ and $E_n$ reversed, producing a polynomial for the indicator of
$E_n$; its complement approximates parity with the same degree.  This gives
the circuit consequence via
Proposition~\ref{prop:partition-chopping}.
\end{proof}

\subsection{Formula approximation}

\begin{theorem}[Top-down formula approximation]
\label{thm:top-down-formula}
There is an absolute constant $C$ with the following property.  Let
$d\ge1$, $m\ge2$, and suppose $\Chop_{d,m}(A,B)$ for nonempty
$A,B\subseteq\{0,1\}^n$.  Then there is a $\Pi_{d-1}$ $\ACzero$ formula $F$ of
fan-in at most $Cm\log m$ and with at most
$\bigl(Cm\log m\bigr)^{d-1}$ leaves such that
\[
 \Pr_A[F=1]\ge\frac{1}{2m}
 \qquad\text{and}\qquad
 \Pr_B[F=1]\le m^{-10}.
\]
\end{theorem}

\begin{proof}
Put $\rho=\frac{1}{2m}$, $\beta=m^{-10}$, and
$r=\lceil2m\ln(\frac{2}{\beta})\rceil$.  As above,
$(1-\rho)^r\le\frac{\beta}{2}$ and $r\beta\le\frac{1}{2}$.
We prove by induction on $d$, directly from the Chopping Game, the stronger
bounds $r$ on fan-in and $r^{d-1}$ on the number of leaves.

A terminal position has an exact separating literal.  For a nonterminal
one-round position, Spoiler's first chop $A'$ makes $(B,A')$ terminal; the
complement of its separating literal is one on $A'$ and zero on $B$.

Let $d\ge2$, and fix Spoiler's first chop $A'$.  Starting from $R_0=B$,
for $j=1,\ldots,r$, stopping if the residual is empty, treat $R_{j-1}$ as a
separate counterfactual response to this fixed chop.
By induction there is a $\Pi_{d-2}$ formula $G_j$ such that
$\Pr_{R_{j-1}}[G_j=1]\ge\rho$ and
$\Pr_{A'}[G_j=1]\le\beta$.  Put
$R_j=R_{j-1}\cap\{G_j=0\}$.  Let $t\le r$ be the number of formulas
constructed.  The final residual $R_t$ has size at most
$\frac{\beta|B|}{2}$.

Let $\overline{G_j}$ denote the formula obtained by interchanging AND and
OR gates and complementing its literals, and set
$F=\bigwedge_{j=1}^t\overline{G_j}$.  Thus $F$ is a $\Pi_{d-1}$ $\ACzero$
formula.
The inputs in $B$ accepted by $F$ are exactly $R_t$.  The union bound gives
$\Pr_{A'}[F=0]\le t\beta\le\frac{1}{2}$.  Therefore
$\Pr_B[F=1]\le\beta$ and
$\Pr_A[F=1]\ge\frac{|A'|}{2|A|}\ge\frac{1}{2m}$.  The top AND has at most $r$
inputs.  Induction gives fan-in at most $r$ throughout and at most
$r^{d-1}$ leaves.  Since $r=O(m\log m)$, this proves the theorem.
\end{proof}

We remark that the covering step is a truncated form of the set-cover
argument underlying Hirahara's duality between exact
$\OR\circ\mathcal C$ formulas and one-sided approximation by $\mathcal C$
\cite{Hirahara2017}.  In the exact-cover setting the usual integrality loss
is logarithmic in the size of the underlying set.  Here we stop after
$O(m\log m)$ stages and allow an $m^{-10}$ residual; this is what makes the
bound independent of $n$.  A related LP-duality observation appears in
\cite[Remark~1]{GoosRiazanovSofronovaSokolov2023}.

Consequently, any lower bound against the weak bounded-depth formula
separators in Theorem~\ref{thm:top-down-formula} gives a Chopping Game lower
bound.  For example, the
formula $F$ in the theorem has average sensitivity
$O((\log m)^d)$ by standard switching-lemma bounds
\cite{Hastad1986}. 
For parity this already gives a
lower bound through the correlation of $F$ with the top Fourier character.
Indeed, with the usual Walsh--Fourier normalization,
$|\widehat F([n])|=\frac12(\Pr_{O_n}[F=1]-\Pr_{E_n}[F=1])\ge\frac1{8m}$.
Since
$\operatorname{as}(F)=4\sum_S|S|\widehat F(S)^2$, it follows that
$\operatorname{as}(F)\ge n/(16m^2)$.  Comparison with the upper bound gives
$m^2(\log m)^d=\Omega(n)$, and hence
$m=\Omega(\sqrt n/(\log n)^{d/2})$ (the conclusion is immediate if $m>n$,
and otherwise $\log m\le\log n$).
This is a bound on the separator $F$, not on the average sensitivity of an
arbitrary function whose input sets admit a Chopping strategy.  By contrast,
Theorem~\ref{thm:top-down-polynomial} applies algebraic approximation one
round at a time in the Chopping Game.  It remains open to find a comparably
direct top-down form of the switching lemma.

Our formulation of the Chopping Game provides a common language for several earlier top-down arguments. In this framework, their different proof constructions correspond to an \textit{explicit} Duplicator strategy built from local limits at successive stages \cite{HastadJuknaPudlak1995,MW19,GoosRiazanovSofronovaSokolov2023}. A principal motivation for this line is
to develop methods not inherited from switching lemmas or polynomial
approximation. The recent depth-four work points to
$\textsc{Inner Product}$ against $\ACzero\circ\XOR$ as one possible target
\cite{GoosRiazanovSofronovaSokolov2023}.
We now consider local KW and Chopping Games, where Spoiler may win without
finding a single global coordinate which separates the two surviving sets.

\section{Local Karchmer--Wigderson and Chopping Games}
\label{sec:distance-k}

Throughout this section, $k$ is fixed while $n$ grows.  We do not assume
that $k$ is odd.

\begin{definition}[$k$-Local and Distance-$k$ KW and Chopping Games]
\label{def:k-local-chopping}\ 
\begin{itemize}
\item
The \emph{$k$-Local KW Game} is the promise version of usual KW Game where Alice receives inputs $x$ and Bob receives input $y$ with the promise that $x,y$ have Hamming distance $\le k$.
\item
The \emph{$k$-Local Chopping Game} has the same moves as the usual Chopping Game, but Spoiler wins upon reaching a position $(A',B')$ where, for every $x\in A'$, there is a coordinate $i=i(x)$
such that
\[
 x_i\ne y_i
 \quad\text{for every $y\in B'$ with $1\le\dist(x,y)\le k$},
\]
and the symmetric condition holds for every $y\in B'$.
\item
We also consider exact distance-$k$ versions of these games, called the {\em Distance-$k$ KW Game} and {\em Distance-$k$ Chopping Game}.
Here the only relevant input pairs $(x,y)$ are those of Hamming distance exactly $k$.
\end{itemize}
\end{definition}

A single coordinate which separates $A$ from $B$ satisfies the $k$-local
condition, and the $k$-local condition implies the distance-$k$ condition.
Thus
\[
 \Chop_{d,m}(A,B)
 \ \Longrightarrow\
 \Chop^{\le k}_{d,m}(A,B)
 \ \Longrightarrow\
 \Chop^{=k}_{d,m}(A,B).
\]
Each implication makes Spoiler's task easier, and consequently makes a lower
bound against the resulting game stronger.

The rest of this section focuses on the exact distance-$k$ problems, rather than the $k$-local problems for $\mr{PARITY}$.
There are two reasons to single out distance exactly $k$.  First, for fixed
$k$ the sphere of radius $k$ contains a $1-O_k(1/n)$ fraction of the ball of
radius $k$.  Second, exact distance retains a key distinction between even and
odd $k$.  Odd-distance pairs have opposite parity and admit the
$2\lceil\log n\rceil$-bit protocol of
Proposition~\ref{prop:odd-distance-upper}; for even $k$, the $\frac{k}{2}\log n$ lower bound in
Theorem~\ref{thm:tarui-even-k} is independent of the number of rounds.

\subsection{Tight bounds for two-round protocols}

A binary linear error-correcting code gives the following two-round upper
bound.

\begin{theorem}[Two-round upper bound]
\label{thm:distance-two-round-upper}
For every fixed $k$ and all sufficiently large $n$, the Distance-$k$ KW
Game has a two-round deterministic protocol of total cost at most
$(k+1)\lceil\log n\rceil$.
\end{theorem}

\begin{proof}
Put $\ell=\lceil\log n\rceil$ and label the coordinates by distinct elements
of $\F_{2^\ell}$, including zero.  The $k$ checks
$\sum_j x_j a_j^r$, for $r=1,3,\ldots,2k-1$, give a $k\ell$-bit syndrome.
If two weight-$k$ vectors had the same syndrome, every odd power sum of their
difference would vanish; the even power sums vanish as Frobenius powers.
Restricted to the nonzero labels, the difference would therefore be a
nonzero word of weight at most $2k$ in the primitive BCH code of designed
distance $2k+1$
\cite{BoseRayChaudhuri60}.  The restriction is nonzero because the difference
of two distinct weight-$k$ vectors has positive even weight.  Thus the
syndrome is injective on vectors of weight exactly $k$.

Alice sends the syndrome of $x$.  From it and $y$, Bob obtains the syndrome
of $x+y$, recovers this weight-$k$ vector, and sends the name of any
coordinate in its support.  Both parties output that coordinate.  The two
messages have lengths $k\ell$ and $\ell$.  The same protocol works with the
roles reversed.
\end{proof}

We first state the set-system consequence used below.  Fix
$r,p\ge2$ and $1\le t<r$.  If
$\mathcal F\subseteq\binom{[n]}r$ contains no $p$-petal sunflower whose
core has size $t$, then
\[
 |\mathcal F|=O_{r,t,p}\bigl(n^{\max\{t,r-t-1\}}\bigr).
\]
F\"uredi's kernel theorem~\cite{Furedi1983} gives a subfamily
$\mathcal F^*\subseteq\mathcal F$, of size at least
$c_{r,p}|\mathcal F|$, in which every intersection of two distinct members
is the core of a $p$-petal sunflower in $\mathcal F$.  Therefore
$\mathcal F^*$ has no two members whose intersection has size $t$.  The
forbidden-intersection theorem of Frankl and F\"uredi
\cite[Theorem~2.1]{FranklFuredi1985} gives
$|\mathcal F^*|=O_r(n^{\max\{t,r-t-1\}})$, proving the claim.

\pagebreak[3]
\begin{lemma}[Fixed-distance estimates]
\label{lem:fixed-distance-estimates}
For every fixed $k\ge1$ and $S\subseteq\{0,1\}^n$:
\begin{enumerate}[(i)]
\item if $S$ contains no pair at distance $2k$, then
$|S|=O_k(\frac{2^n}{n^k})$;
\item if $k$ is even and $S$ contains no three points at pairwise distance
$k$, then $|S|=O_k(\frac{2^n}{n^{k/2}})$.
\end{enumerate}
\end{lemma}

\begin{proof}
For $z\in\F_2^n$ and an integer $r$, let
\[
 \mathcal F_r(z)=
 \left\{I\in\binom{[n]}r:z+\1_I\in S\right\}.
\]
Averaging over $z$, or simply double-counting the pairs $(z,I)$, gives
\[
 \E_z|\mathcal F_r(z)|=
 \frac{|S|}{2^n}\binom nr.
\]

For (i), take $r=2k$.  If
$I,J\in\mathcal F_{2k}(z)$ satisfy $|I\cap J|=k$, then the two points
$z+\1_I$ and $z+\1_J$ are at distance
$|I\mathbin\triangle J|=2k$, contrary to the hypothesis.  Thus
$\mathcal F_{2k}(z)$ contains no two-petal sunflower with a core of size
$k$.  By the preceding bound, it has $O_k(n^k)$ members for every $z$.
The averaging identity, together with
$\binom n{2k}=\Theta_k(n^{2k})$, gives
$|S|=O_k(2^n/n^k)$.

For (ii), write $k=2h$ and take $r=2h$.  A three-petal sunflower
$I_1,I_2,I_3\in\mathcal F_{2h}(z)$ with core of size $h$ would satisfy
$|I_a\mathbin\triangle I_b|=2h=k$ whenever $a\ne b$.  The three points
$z+\1_{I_1},z+\1_{I_2},z+\1_{I_3}$ would therefore be pairwise at distance
$k$, contrary to the hypothesis.  Hence the same bound gives
$|\mathcal F_{2h}(z)|=O_h(n^h)$ for every $z$.  Applying the averaging
identity once more yields
$|S|=O_h(2^n/n^h)=O_k(2^n/n^{k/2})$.
\end{proof}

\begin{theorem}[Two-round lower bound]
\label{thm:distance-two-round-lower}
For every fixed $k$, every two-round deterministic protocol for the
Distance-$k$ KW Game has total cost at least
$(k+1)\log n-O_k(1)$.
\end{theorem}

\begin{proof}
For $n<2k$, the claim is absorbed by the $O_k(1)$ term.  Assume $n\ge2k$.
Let $T$ be the total cost and suppose Alice speaks first.  Her first message
has at most $2^T$ possible values, so some message class
$S\subseteq\{0,1\}^n$ has size at least $2^{n-T}$.  We claim that $S$
contains no two strings $x,x'$ at distance $2k$.  Otherwise choose a common
midpoint $y$, at distance $k$ from each.  Bob has the same input and sees the
same first message on $(x,y)$ and $(x',y)$, so he sends the same reply and
produces the same output $i$ in both cases.  Correctness would require
$y_i$ to differ from both $x_i$ and $x'_i$.  This is impossible, since at
every coordinate $y$ agrees with at least one of $x,x'$.

Lemma~\ref{lem:fixed-distance-estimates}(i) gives
$2^{n-T}=O_k(\frac{2^n}{n^k})$ and hence $k\log n - O_k(1)$ for Alice's communication alone. The addition $\log n$ is required for Bob to communicate a single coordinate where $x_i \neq y_i$.
\end{proof}

\subsection{Odd vs even $k$}

Karchmer's familiar diagonal argument for the Universal Relation has an
exact-distance analogue; see \cite{TardosZwick1997}.  Complete a
protocol for $\UR_n$ off its promise.  If three distinct diagonal inputs
reach the same transcript, the rectangle property puts all six ordered cross
pairs there.  Correctness then forces the six local output values to be one
coordinate which differs on each of the three unordered pairs, an
impossibility.  Thus a diagonal transcript class has size at most two, giving
the round-independent total-cost lower bound $n-1$.  At fixed distance, the
same argument says that a diagonal transcript class must be triangle-free in
the exact-distance graph.

\ifnum\submissionmode=1

\else
The second author learned the diagonal argument in the proof of the following
theorem from Jun Tarui (personal communication).  The sharp error term stated
here uses Lemma~\ref{lem:fixed-distance-estimates}(ii).
\fi

\begin{theorem}
\label{thm:tarui-even-k}
For every fixed even $k\ge2$, every deterministic protocol for the Distance-$k$
KW Game, with an arbitrary number of rounds, has total cost at least
$\frac{k}{2}\log n-O_k(1)$.  Consequently, for every $d\ge1$, a $d$-round
protocol has a message of length at least
$\frac{k}{2d}\log n-O_k(1)$.
\end{theorem}

\begin{proof}
Let $T$ be the total cost.  Complete the protocol arbitrarily off its promise
without increasing its cost, and view it as a binary tree of depth at most
$T$.  Among the at most $2^T$ transcripts on diagonal inputs $(x,x)$, one
is shared by a set $S$ of at least $2^{n-T}$ strings.

Suppose that $x,y,z\in S$ have pairwise distance $k$.  The rectangle
property puts all six ordered cross pairs in the same transcript.  If
$a(u)$ and $b(v)$ are the two output functions at this transcript,
correctness on these cross pairs forces all six values
$a(x),a(y),a(z),b(x),b(y),b(z)$ to be one common coordinate $i$.  That
coordinate would have to differ on each of the three pairs $xy,yz,zx$,
which is impossible for the three bits $x_i,y_i,z_i$.  Thus $S$ contains no
such triangle.  Lemma~\ref{lem:fixed-distance-estimates}(ii) gives
$2^{n-T}=O_k(\frac{2^n}{n^{k/2}})$, proving the total-cost lower bound.  The
per-message conclusion follows by averaging over the $d$ messages.
\end{proof}

\begin{proposition}[Odd-distance upper bound]
\label{prop:odd-distance-upper}
For every odd $k$, the Distance-$k$ KW Game has a deterministic protocol of
total cost at most $2\lceil\log n\rceil$.
\end{proposition}

\begin{proof}
Pad the inputs by common zeros so that their length is a power of two.
Maintain a block on which the two restricted strings have opposite parity.
Alice sends the parity of her string on the first half of the block.  Bob
compares it with his parity on that half and sends one bit selecting a half
on which the two parities differ.  After $\lceil\log n\rceil$ stages, the
remaining block is a coordinate on which the inputs differ, known to both
parties.  This uses two one-bit rounds per stage, hence
$2\lceil\log n\rceil$ rounds and bits in total.  It is the usual binary-search
protocol for the odd-parity part of the Universal Relation.
\end{proof}

\begin{remark}[Even and odd distance]
\label{rem:even-odd-distance}
Theorem~\ref{thm:tarui-even-k} is independent of the number of rounds.  For
every fixed even $k>4$, its leading term $\frac{k}{2}\log n$ is larger than the
$2\log n+O(1)$ unrestricted-round upper bound at every odd distance.  There is
no direct transfer between the two arguments: odd distance is a parity
restriction, whereas even-distance pairs lie within the same parity class.
\end{remark}

\subsection{Distance-$k$ Conjectures}

As one of the main conceptual contributions of this work, we advance the following conjecture, which can be seen as a $k$-local version of the classical theorem that $O(1)$-depth circuits for $\PARITY$ require size $n^{\omega(1)}$.

\begin{conjecture}[Distance-$k$ KW Conjecture]
\label{conj:distance-k}
For every $d\ge2$ and $c>0$, there is $k_0=k_0(d,c)$ such that, for every
fixed $k\ge k_0$ and all sufficiently large $n$, the existence of a
$d$-round $m$-ary protocol in the Distance-$k$ KW Game implies $m\ge n^c$.
Equivalently, the conjecture asserts $m=n^{\omega(1)}$ in the regime
$d\ll k\ll n$.
\end{conjecture}

We also state a (top-down) Chopping Game version of this conjecture, which is a formally stronger statement.

\begin{conjecture}[Distance-$k$ Chopping Conjecture]
\label{conj:distance-k-chopping}
For every $d\ge1$ and $c>0$, there is $k_0=k_0(d,c)$ such that, for every
fixed $k\ge k_0$ and all sufficiently large $n$,
\[
 \Chop^{=k}_{d,m}(\{0,1\}^n,\{0,1\}^n)
 \quad\Longrightarrow\quad
 m\ge n^c.
\]
Equivalently, in abbreviated form, the conjecture asserts
$m=n^{\omega(1)}$ in the regime $d\ll k\ll n$.
\end{conjecture}

Conjecture \ref{conj:distance-k} (and its stronger chopping version \ref{conj:distance-k-chopping}) qualitatively strengthen the classical lower bound $\PARITY\notin\ACzero$.  Indeed, for
any fixed depth $d$ and exponent $c$, choose a sufficiently
large odd $k \ge k_0(d,c+1)$.  A depth-$d$, fan-in-$n^c$ circuit for $\PARITY$ would yield a KW protocol in the KW Game, hence also in the Distance-$k$ KW Game and the Distance-$k$ Chopping Game,
contradicting the conjecture.

\section{Lower Bound for Linear Protocols}
\label{sec:linear-distance-k}

A linear message partitions the active input set into intersections with
affine subspaces.  Thus every linear KW protocol gives an affine Partition
Game strategy and, by retaining a largest part, an affine Chopping strategy.
The known upper bounds below already have these forms.  The code
used in Theorem~\ref{thm:distance-two-round-upper} has a linear
two-round variant: Alice and Bob each send the
$k\lceil\log n\rceil$-bit syndrome of their input, after which both recover
$x+y$.  
The standard depth-$d$ circuits for PARITY 
give a $d$-round affine Chopping strategy with
$\log m=O(n^{1/(d-1)})$ for $d\ge2$.  At each step the construction fixes a
pattern of block parities, and hence passes to an affine subspace.  Its final
coordinate separates the two surviving sets globally, so it also wins every
$k$-local version.

We prove both lower bounds by reducing to the exact-weight $\FtwoELEMX$ query
problem.
Fixing Alice's input to $0^n$ turns a $d$-round linear KW protocol into an
algorithm with $q=d-1$ query rounds; some of these query batches may be
empty.  An affine Chopping strategy gives a different reduction, in which
Duplicator keeps the two sides related by an unknown translation and one
query batch simulates each chop.  The resulting exponents are
$k^{1/(d-1)}$ for linear protocols and $k^{1/d}$ for affine strategies,
rather than the linear dependence on $k$ predicted by the conjectures.  The
exact statements first assume that $k$ is the corresponding perfect power.
For the $k$-local affine consequence, a nearby odd perfect power removes
this restriction at the cost of a constant factor in the exponent.

\begin{definition}[Linear KW protocol]
A KW protocol is $\F_2$-linear if, at every transcript $\tau$, its next
message has the form $U_\tau x$ or $U_\tau y$ for a matrix $U_\tau$ over
$\F_2$, according to which party speaks.  It is $L$-bit if every such
matrix has at most $L$ rows.
\end{definition}

We measure the number of bits actually sent, rather than the size of the
message's range on the residual input set.  Accordingly, after padding
matrices with zero rows, an $L$-bit message ranges over $\{0,1\}^L$, so
$m=2^L$.

The \textsc{Element Extraction} ($\ELEMX$) query problem of Chakrabarti and
Stoeckl~\cite{ChakrabartiStoeckl2021} asks for a coordinate in the support of
an unknown nonzero Boolean vector using linear queries.  We use its
exact-weight restriction over $\F_2$.

\begin{definition}[$\FtwoELEMX$]
\label{def:f2-elemx}
In $\FtwoELEMX(k,n)$ the unknown vector $z\in\F_2^n$ has Hamming weight
$k$.  In each query round, the algorithm chooses a batch of vectors
$a\in\F_2^n$, as a function of the answers in earlier rounds, and receives
$\langle a,z\rangle$ for each.  It must output an $i\in\supp(z)$.  The size
of a batch is its number of queries.  A $q$-round algorithm uses at most
$q$ batches; empty batches are allowed.
\end{definition}

We write $\FtwoELEMX(\Odd,n)$ for the variant in which $z$ may have any
odd Hamming weight. 
The search lower
bound in ~\cite[Section~5.2]{Rossman2017Subspace}
with $V\setminus U$ specialized to the even-weight
subspace $E_n\subseteq\F_2^n$ implies that every
deterministic $q$-round algorithm for $\FtwoELEMX(\Odd,n)$ has a batch
of size at least $n^{1/q}-1$. That linear-algebraic argument uses the full
odd-weight promise, and we do not know how to extend it to the exact-weight
problem $\FtwoELEMX(k,n)$. Chakrabarti and Stoeckl give a round-elimination
proof for the odd-weight problem.  The exact-weight promise here requires the
fixed-core sunflower lemma below.

We reserve $q$ for the number of query rounds.  In the reduction from a
$d$-round linear KW protocol below, $q=d-1$: after fixing Alice's input, her
first message is known, and each remaining communication round is represented
by one query batch, possibly empty.

\begin{lemma}[Fixed-core sunflower bound
  {\cite[Theorem~1]{BradacBucicSudakov2023}}]
\label{lem:fixed-core-sunflower}
For every fixed $\ell\ge2$ there is a constant $C_\ell$ such that, for every
integer $t\ge2$, every family
$\mathcal F\subseteq\binom{[N]}{2\ell}$ with more than
$C_\ell N^\ell t^\ell$ members contains $t$ sets of the form
$C\mathbin{\dot\cup}P_1,\ldots,C\mathbin{\dot\cup}P_t$, where
$|C|=|P_1|=\cdots=|P_t|=\ell$ and the petals $P_1,\ldots,P_t$ are pairwise
disjoint.
\end{lemma}

The lower bound follows by round elimination.  Write $k=r^q$.  For every
batch of $L$ queries, some syndrome is shared by many disjoint $r$-sets.
Replacing each new coordinate by one such set makes the first-round answers
constant, reduces the weight from $r^s$ to $r^{s-1}$, and leaves a universe
of size roughly $N2^{-L/r}$.  The iteration either fails early or ends with
two disjoint supports having the same final transcript; either outcome forces
$qL\ge r\log n-O_k(q)$.

\begin{theorem}[Lower bound for $\FtwoELEMX$]
\label{thm:f2-elemx-round-lower}
Let $q\ge1$, and assume that $k^{1/q}$ is an integer at least $2$.  Every
deterministic $q$-round algorithm for $\FtwoELEMX(k,n)$ has a batch of size
at least
\[
 \frac{k^{1/q}}{q}\log n-O_k(q).
\]
\end{theorem}

\begin{proof}
Put $r=k^{1/q}$.  We first record the one consequence of
Lemma~\ref{lem:fixed-core-sunflower} that will be used repeatedly.  Let
$U$ be a matrix with at most $L$ rows and $N$ columns.  Among the
$2r$-subsets of $[N]$, some fiber of the map $S\mapsto U\mathbf 1_S$ has
at least $\binom{N}{2r}2^{-L}$ members.  Choose $c_r>0$ sufficiently small,
depending only on $r$.  Applying the lemma to this fiber and cancelling its
common core shows that there are
\[
 M=\left\lfloor c_rN2^{-L/r}\right\rfloor
\]
pairwise disjoint $r$-sets $P_1,\ldots,P_M$ such that
$U\mathbf 1_{P_1}=\cdots=U\mathbf 1_{P_M}$, provided $M\ge2$.

Suppose that every batch of a $q$-round algorithm has size at most $L$.
Consider a stage at which the hidden vector has weight $r^s$ on a universe
of size $N$, with $s\ge2$ rounds remaining.  Apply the preceding observation
to the first query matrix.  If $M\ge r^{s-1}$, restrict the hidden vectors
to those of the form
\[
 z(W)=\sum_{j\in W}\mathbf 1_{P_j},
 \qquad W\in\binom{[M]}{r^{s-1}}.
\]
Their first-round answers are all the same.  Every later query remains a
linear query in $\mathbf 1_W$, and an output coordinate in $\supp(z(W))$
belongs to a unique petal and hence identifies an element of $W$.  We have
therefore eliminated the first round, obtaining an $(s-1)$-round algorithm
for $\FtwoELEMX(r^{s-1},M)$.

Repeat this reduction for the first $q-1$ rounds.  At every successful
step the new universe size satisfies
\[
 N'\ge \gamma_rN2^{-L/r}
\]
for a constant $\gamma_r>0$.  If the reduction first fails after $j$ steps,
where $0\le j\le q-2$, because there are fewer than the required number of
petals, then
\[
 N\ge\gamma_r^j n2^{-jL/r}
 \quad\text{and}\quad
 c_rN2^{-L/r}<r^{q-j-1}+1.
\]
Since $r\le k$ and the constants $c_r,\gamma_r$ depend only on $r$, the
loss accumulated in at most $q$ reductions is $O_k(q)$.  Consequently,
\[
 (j+1)L\ge r\log n-O_k(q).
\]
As $j+1\le q$, this implies
\[
 L\ge \frac rq\log n-O_k(q),
\]
and the desired bound follows.

It remains to consider the case in which all $q-1$ reductions succeed.
They leave a one-round algorithm for weight $r$ on a universe of size
\[
 N\ge\gamma_r^{q-1}n2^{-(q-1)L/r}.
\]
If $\lfloor c_rN2^{-L/r}\rfloor\ge2$, there are two disjoint $r$-sets
with the same final transcript.  The algorithm would have to give the same
answer on both, which is impossible.  Thus $c_rN2^{-L/r}<2$, and the last
display gives
\[
 qL\ge r\log n-O_k(q).
\]
After division by $q$, this is stronger than the stated bound.
\end{proof}

\begin{theorem}[Linear protocols]
\label{thm:linear-distance}
Let $d\ge2$, and assume that $k^{1/(d-1)}$ is an integer at least $2$.
Every $d$-round linear protocol for the Distance-$k$ KW Game has a message
of length at least
\[
 \frac{k^{1/(d-1)}}{d-1}\log n-O_k(1).
\]
In particular, the same expression is a lower bound on total
communication.
\end{theorem}

\begin{proof}
Let the protocol use at most $L$ bits in each round.  On an input
$z\in\F_2^n$ of weight $k$, run it on the promised pair $(x,y)=(0^n,z)$.
Alice's first message is fixed.  Thereafter each message of Bob is a batch
of at most $L$ linear queries in $z$, while each message of Alice is known
from the transcript.  Alice's final output depends only on the transcript
and her fixed input.  By agreement and correctness of the protocol, it lies
in $\supp(z)$.  Regard each of the remaining $d-1$ communication rounds as
one query round, using an empty batch whenever Alice speaks.  This gives a
$(d-1)$-round algorithm for $\FtwoELEMX(k,n)$ with at most $L$ queries in
every batch.  Applying Theorem~\ref{thm:f2-elemx-round-lower} with $q=d-1$
gives the result.  Indeed, if $k=r^{d-1}$ with $r\ge2$, then
$d-1\le\log k$, so the resulting $O_k(d-1)$ term is $O_k(1)$.
\end{proof}

\begin{definition}[Affine Spoiler]
A Spoiler strategy is \emph{affine} if every chop of the active side $A$
has the form $A'=A\cap(a+V)$ for an affine subspace
$a+V\subseteq\F_2^n$.
\end{definition}

For affine Spoiler strategies, we reduce directly to $\FtwoELEMX$, without
using a KW protocol.  Given a hidden vector $h$, Duplicator maintains two
affine sets related by translation by $h$.  A chop of relative codimension
$s$ requires only $s$ linear queries to determine the corresponding
translate on the other side.  Thus one query batch simulates each Chopping
round.

\begin{theorem}[Affine Spoiler]
\label{thm:affine-spoiler}
Let $d\ge1$, and assume that $k^{1/d}$ is an integer at least $2$.  If an
affine Spoiler strategy witnesses
$\Chop^{=k}_{d,m}(\{0,1\}^n,\{0,1\}^n)$, then
\[
 \log m\ge \frac{k^{1/d}}{d}\log n-O_k(1).
\]
If $k$ is odd, the same conclusion holds when
$\Chop^{=k}_{d,m}(\{0,1\}^n,\{0,1\}^n)$ is replaced by
$\Chop^{\le k}_{d,m}(O_n,E_n)$.
\end{theorem}

\begin{proof}
Let $h\in\F_2^n$ be the hidden vector in $\FtwoELEMX(k,n)$.  We simulate
the affine strategy while choosing Duplicator's responses so that the two
surviving sets are translates by $h$.  Thus the current ordered position
has the form
\[
 A=a+V,
 \qquad
 B=A+h.
\]
This holds initially for the two copies of $\{0,1\}^n$.  It also holds for
$(O_n,E_n)$ when $k$ is odd, since translation by $h$ swaps parity.

Suppose that Spoiler replaces $A$ by the affine subcoset $A'=a'+W$.  Its
relative density in $A$ is $2^{-s}$, where
$s=\operatorname{codim}_V(W)$, and legality gives $s\le\log m$.  The
simulation already knows the coset $h+V$.  Choose linear functionals
$\lambda_1,\ldots,\lambda_s\in W^\perp$ whose images form a basis of
$W^\perp/V^\perp$, and query
$\langle\lambda_1,h\rangle,\ldots,\langle\lambda_s,h\rangle$.  Together
with the values of the functionals in $V^\perp$, which are determined by
$h+V$, these answers determine every functional in $W^\perp$ and hence the
coset $h+W$.  Duplicator may therefore respond with $B'=A'+h$.  After the
order is reversed, the new position is $(B',A')$; since $A'=B'+h$ over
$\F_2$, the same translate invariant continues.

When Spoiler's stopping condition holds, choose the least $x$ on the first
side.  The point $x+h$ lies on the opposite side and is at distance $k$
from $x$.  Under either stopping rule, a witnessing coordinate $i$ for
$x$ satisfies $x_i\ne(x+h)_i$, and hence $h_i=1$.  Take the least such $i$.
We have constructed an algorithm for $\FtwoELEMX(k,n)$ using at most $d$
query rounds and at most $\log m$ queries in each batch.  Applying
Theorem~\ref{thm:f2-elemx-round-lower} with $q=d$ gives an $O_k(d)$
remainder.  Since $k=r^d$ with $r\ge2$, we have $d\le\log k$, and this
remainder is $O_k(1)$.
\end{proof}

\begin{corollary}[Affine $k$-local games]
\label{cor:affine-spoiler-all-k}
For every $d\ge1$ and $k\ge3^d$, and all sufficiently large $n$, an affine
winning strategy in the $d$-round $k$-local $m$-Chopping Game for parity
requires $m\ge n^{\Omega(k^{1/d}/d)}$.
\end{corollary}

\begin{proof}
Let $r$ be the largest odd integer such that $r^d\le k$.  Then
$r\ge\frac{k^{1/d}}{3}$.  A winning strategy for the $k$-local game is also a
winning strategy for the $r^d$-local game.  Apply
Theorem~\ref{thm:affine-spoiler} with the odd distance $r^d$.
\end{proof}

\section{Graph-Theoretic Conjecture}
\label{sec:graph-analogue}

The preceding lower bounds use the vector-space structure of $\F_2^n$.  We
now consider a version of the Distance-$k$ question in which this structure
is absent.  The graph conjecture below is not formally easier than the
$k$-local Chopping Conjecture, but it isolates a simpler local problem.  A
lower-bound method for this problem may lead to fundamentally new circuit
lower-bound techniques.

Let $G$ be a connected graph and let $k\ge1$.

\begin{definition}[Graph Distance-$k$ Relation]
\label{def:graph-distance-relation}
Alice receives $x\in V(G)$ and Bob receives $y\in V(G)$ with
$\dist_G(x,y)=k$.  Alice outputs a neighbor $x'$ of $x$ such that
$\dist_G(x',y)=k-1$, and Bob outputs a neighbor $y'$ of $y$ such that
$\dist_G(y',x)=k-1$.  The two outputs are separate local directions; a
protocol for this relation need not give them a common label.
\end{definition}

In $Q_n$, the coordinate labels globally identify parallel edges, and the
Distance-$k$ KW Game requires the parties to choose the same label.  A
general regular graph has no intrinsic identification between edges at
different vertices.  The relation above discards this requirement and is
therefore a relaxation of the literal edge-labeled analogue, but it depends
only on the graph metric.

\begin{definition}[Graph Distance-$k$ Chopping Game]
\label{def:graph-distance-k-chopping}
The moves are those of Definition~\ref{def:chopping-game}, with role-labeled
sets $A,B\subseteq V(G)$.  Spoiler may stop and win from $(A,B)$ if, for
every $x\in A$, there is a neighbor $x'(x)$ satisfying
$\dist_G(x'(x),y)=k-1$ for every $y\in B$ at distance $k$ from $x$, and,
symmetrically, for every $y\in B$ there is a neighbor $y'(y)$ satisfying
$\dist_G(y'(y),x)=k-1$ for every $x\in A$ at distance $k$ from $y$.  We
write $\Chop^{G,=k}_{d,m}(A,B)$ when Spoiler can force this condition within
$d$ rounds.
\end{definition}

This is the zero-message rule for the graph relation.  The largest-part
argument therefore turns an $m$-ary $d$-round protocol into a winning graph
Chopping strategy with the same parameters, starting from two role-labeled
copies of $V(G)$.

\subsection{Local hypercubes}

The same local games may be played on graphs whose radius-$k$ neighborhoods
are isomorphic to those of the cube.  Let $Q_n$ denote the
$n$-dimensional hypercube graph on $\{0,1\}^n$.

\begin{definition}
\label{def:k-local-n-hypercube}
A finite connected $n$-regular graph $G$ is a
\emph{$k$-local $n$-hypercube} if, for every $v\in V(G)$, the subgraph
induced by the radius-$k$ ball about $v$ is isomorphic, as a rooted graph,
to the radius-$k$ ball about $0^n$ in $Q_n$.
\end{definition}

These are the $Q_n^k$-like graphs introduced by Klav\v{z}ar, Koolen, and
Mulder \cite[Theorem~2.1]{KlavzarKoolenMulder1999}.  For $k\ge3$, the finite
connected examples are exactly the normal quotients $(Q_n)_K$ (the graphs
on the $K$-orbits) by subgroups
$K\le\Aut(Q_n)$ whose nonidentity elements move every vertex by at least
$2k+2$ \cite[Theorem~1.2 and Corollary~1.3]{Fawcett2016}.  Translation
subgroups give Cayley quotients, but general $K$ may also permute
coordinates, and the quotient need not even be vertex-transitive
\cite[Example~3.12]{Fawcett2016}.  We use only the following Cayley
construction.

\begin{proposition}[Cayley quotients]
\label{prop:cayley-quotients}
Let $V$ be a finite-dimensional vector space over $\F_2$, and let
$g_1,\ldots,g_n\in V$ span $V$ and be $(2k+1)$-wise linearly independent,
meaning that every nonempty sum of at most $2k+1$ of the generators is
nonzero.  Then $\operatorname{Cay}(V,\{g_1,\ldots,g_n\})$ is a connected
$k$-local $n$-hypercube.
\end{proposition}

\begin{proof}
The generators are distinct and nonzero, so the Cayley graph is $n$-regular;
spanning gives connectedness.  By translation it is enough to consider the
ball about $0$.  Every vertex in this ball has the form
$g_I=\sum_{i\in I}g_i$ for some $I\subseteq[n]$ with $|I|\le k$.  If
$g_I=g_J$, then $I\mathbin\triangle J$ gives a linear dependence of size at
most $2k$, so $I=J$.  If $g_I$ and $g_J$ are adjacent, then
$g_I+g_J=g_\ell$ for some $\ell\in[n]$.  The set
$(I\mathbin\triangle J)\mathbin\triangle\{\ell\}$ has size at most $2k+1$
and sums to zero.  It must therefore be empty, which says that
$I\mathbin\triangle J=\{\ell\}$.  The converse is immediate.  Thus
$I\mapsto g_I$ identifies the induced ball with the radius-$k$ ball in
$Q_n$.
\end{proof}

\subsection{High-girth regular graphs}

Local hypercubes retain the full radius-$k$ geometry of the cube.  The
high-girth model retains only the uniqueness of length-$k$ geodesics.  Let
$G$ be a finite connected $n$-regular graph with $\girth(G)>2k$.  Two
vertices at distance $k$ have a unique shortest path, since two such paths
would contain a cycle of length at most $2k$.

\begin{proposition}[Two-round coloring protocol]
\label{prop:graph-coloring-protocol}
Suppose that $\kappa:V(G)\to[M]$ assigns different colors to the endpoints
of every path of length $2k$.  Then the graph Distance-$k$ relation has a
two-round $M$-ary protocol.  Such a coloring exists with
\[
 M\le n(n-1)^{2k-1}+1.
\]
Consequently, every message has length at most $2k\log n+O(1)$ bits.
\end{proposition}

\begin{proof}
Alice sends $\kappa(x)$.  Consider the vertices $z$ at distance $k$ from
$y$ with $\kappa(z)=\kappa(x)$.  The geodesics from $y$ to all such $z$
begin with the same edge.  Indeed, the girth assumption makes two such
geodesics which leave $y$ in different directions disjoint away from $y$;
together they would form a path of length $2k$ whose endpoints have the same
color.  Bob therefore knows his required neighbor and outputs it.  He sends
$\kappa(y)$, and the same argument lets Alice output her required neighbor.

To obtain the coloring, join two vertices when they are endpoints of a
length-$2k$ path.  This graph has maximum degree at most
$n(n-1)^{2k-1}$, the number of nonbacktracking walks of length $2k$ from
one vertex, and hence has a proper coloring with the stated number of
colors.
\end{proof}

Finite regular graphs of any prescribed degree and girth exist
\cite{Sachs1963}.  In the conjecture below, $n$ denotes the degree and the
order $|V(G)|$ is unrestricted.  Proposition~\ref{prop:graph-coloring-protocol}
shows that this does not introduce an artificial encoding cost: the
two-round upper bound depends only on the radius-$2k$ geometry and is
$O(k\log n)$ bits per round.

\begin{conjecture}[High-girth Distance-$k$ Conjecture]
\label{conj:high-girth-distance}
For every $d\ge2$ and $k\ge1$, and all sufficiently large $n$, there is a
finite connected $n$-regular graph $G$ with $\girth(G)>2k$ for which every
$d$-round protocol for the graph Distance-$k$ relation has maximum message
length $\Omega_d(k\log n)$.
\end{conjecture}

The degree $n$ takes the place of the cube dimension, while the message
alphabet takes the place of fan-in.  In this sense,
Conjecture~\ref{conj:high-girth-distance} is a graph-theoretic version of
$\PARITY\notin\ACzero$.  The formally stronger Chopping version asks, for
the same graphs, that
$\Chop^{G,=k}_{d,m}(V(G),V(G))$ imply $m\ge n^{\Omega_d(k)}$.

\subsection{The infinite regular tree}

Finiteness is essential in
Conjecture~\ref{conj:high-girth-distance}.  The rooted infinite tree behaves
differently: a public root supplies a parent relation, and with this extra
structure the Distance-$k$ relation has a protocol with only logarithmic
dependence on $k$.  Let $T_n$ be the infinite $n$-regular tree, where
$n\ge2$.  Fix a public root $o$ and a public bijection
$\operatorname{port}_v:N_{T_n}(v)\to[n]$ at each vertex $v$.

\begin{theorem}[Infinite-tree upper bound]
\label{thm:infinite-tree-upper}
For every $k\ge1$, the Distance-$k$ relation on $T_n$ has a deterministic
three-round protocol of prefix-free total cost at most
\[
 \lceil\log n\rceil+\lceil\log(2k+1)\rceil
   +\lceil\log(k+1)\rceil
 =\log n+O(\log k).
\]
Along every transcript, at most one edge port is sent.
\end{theorem}

\begin{proof}
For $v\in V(T_n)$, let $h(v)=\dist_{T_n}(o,v)$.  Suppose that Alice and Bob
receive $x$ and $y$ at distance $k$, and let $w$ be the last common vertex
of the $o$--$x$ and $o$--$y$ paths.  Put $r=\dist_{T_n}(x,w)$ and
$s=\dist_{T_n}(y,w)$.  Thus $r+s=k$ and $h(y)-h(x)=s-r$.

Alice first sends $h(x)$ modulo $2k+1$.  Since $|h(y)-h(x)|\le k$, Bob
recovers $\delta=h(y)-h(x)$ as the unique integer in $[-k,k]$ compatible
with this residue.  He computes $r=\frac{k-\delta}{2}$ and $s=k-r$, and sends
$r\in\{0,\ldots,k\}$ to Alice.  If $r=0$, then $x$ is the $k$th ancestor
of $y$.  Bob can therefore determine $x$ from his own input, and he appends
the port at $x$ of the first edge toward $y$.

Alice now knows $r$ and $s$.  If $s=0$, then $y$ is the $k$th ancestor of
$x$, so Alice sends the port at $y$ of the first edge toward $x$.  In every
other case her third message is empty.  If $r>0$, Alice's required neighbor
is the parent of $x$; if $s>0$, Bob's required neighbor is the parent of
$y$.  The two exceptional cases $r=0$ and $s=0$ are handled by the port
sent by the other player.  Since $r+s=k>0$, at most one exception occurs.

The residue and the value of $r$ use fixed-length encodings.  Bob's possible
messages are the codeword for $r>0$, or the codeword for $r=0$ followed by a
fixed-length port.  These messages are prefix-free.  Whether Alice's last
message is empty or a fixed-length port is determined by the transcript.
On an off-promise input, any undefined quantity or ancestor is replaced by a
fixed default of the prescribed length.  Such branches do not occur on
promised inputs.  This gives the claimed total cost.
\end{proof}

\begin{remark}
\label{rmk:infinitary-chopping}
Note that this protocol in $T_n$ uses the common ancestor and parent relations
supplied by the distinguished root.  In a finite high-girth graph, the
relevant geodesic need not lie in any fixed spanning tree, so the preceding
reconstruction from heights fails.  The theorem therefore gives no upper
bound for the finite graphs in
Conjecture~\ref{conj:high-girth-distance}.

One may nevertheless formulate an automorphism-invariant Chopping Game on
$T_n$.  A position may be represented by an automorphism-invariant random
vertex marking arising as a local limit of finite marked regular graphs.
Conditional probabilities at a random root replace relative cardinalities,
and legal moves must commute with rerooting.  Such a formulation has
no globally distinguished vertex or parent map and therefore excludes the
protocol above.
\end{remark}

\section{Future Directions}
\label{sec:future-directions}

The main open problem is the Distance-$k$ KW Conjecture
(Conjecture~\ref{conj:distance-k}): prove $m=n^{\omega(1)}$ in the regime
$d\ll k\ll n$.
Conjecture~\ref{conj:distance-k-chopping} makes the same qualitative claim
for the more permissive Distance-$k$ Chopping Game.  The latter is a qualitative
strengthening of $\PARITY\notin\ACzero$, rather than merely another
formulation of the usual circuit lower bound.

The results of Section~\ref{sec:linear-distance-k} give exponents growing
with $k$ for linear protocols and affine Spoiler strategies.  Both arguments
use the vector-space structure of the cube: linear messages become batches
of queries in $\FtwoELEMX$, while affine chops preserve two sets related by
an unknown translation.  The main open problems concern unrestricted
protocols and Chopping strategies; on the Chopping side, this means removing
the affine hypothesis from our lower bound.  It would already suffice to
prove $m\ge n^{\gamma_d(k)}$ for some $\gamma_d(k)\to\infty$.  The
approximation argument of Section~\ref{sec:top-down-approximation} does not
give such a bound in the fixed-$k$ regime.

The high-girth relation removes the coordinate labels and vector-space
structure of the cube, retaining finite branching and a unique short
geodesic.  It is not a formal reduction of the cube conjectures, but it
provides a simplified setting for developing local lower-bound methods.
Proposition~\ref{prop:graph-coloring-protocol} identifies the natural
$k\log n$ scale, while Theorem~\ref{thm:infinite-tree-upper} gives a much
cheaper protocol on the rooted, port-labeled infinite tree.  As explained
there, its reconstruction from the parent relation does not transfer to a
finite high-girth graph.  A method which proves the high-girth conjecture and
also applies to the cube would give the kind of new lower-bound technique
sought here.

\ifnum\submissionmode=0
\subsection*{Acknowledgments}
G. K. is supported by the Templeton Foundation, Inc.\ (funder DOI: 501100011730) under the grant DOI:10.54224/20650
no.\ 20650 on Building Diverse Intelligences through
Compositionality and Mechanism Design, and additionally by the J. Grochow and R. Frongillo start-up funds at the University of Colorado Boulder.
The authors thank Josh Grochow for helpful conversations throughout all stages of this project.

A substantial portion of this work was carried out during a visit of G. K.\ to Duke University in Fall 2025.
\fi

\subsection*{AI disclosure}

All definitions, results, and conjectures in the paper were developed, proved, and formulated by the author(s).  AI tools were used for copy-editing tasks only.

\begingroup
\renewcommand{\H}[1]{\accent"7D #1}
\newcommand{\etalchar}[1]{$^{#1}$}

\endgroup


\begin{thebibliography}{ABC{\etalchar{+}}13}

\bibitem[Ajt83]{Ajtai1983}
Mikl{\'o}s Ajtai.
\newblock $\Sigma^1_1$-formulae on finite structures.
\newblock {\em Annals of Pure and Applied Logic}, 24(1):1--48, 1983.

\bibitem[BBS23]{BradacBucicSudakov2023}
Domagoj Brada{\v{c}}, Matija Buci{\'c}, and Benny Sudakov.
\newblock Tur{\'a}n numbers of sunflowers.
\newblock {\em Proceedings of the American Mathematical Society},
  151(3):961--975, 2023.

\bibitem[BRC60]{BoseRayChaudhuri60}
R.~C. Bose and D.~K. Ray-Chaudhuri.
\newblock On a class of error correcting binary group codes.
\newblock {\em Information and Control}, 3(1):68--79, 1960.

\bibitem[Bra10]{Braverman2010}
Mark Braverman.
\newblock Polylogarithmic independence fools {$\ACzero$} circuits.
\newblock {\em Journal of the ACM}, 57(5):28:1--28:10, 2010.

\bibitem[CS21]{ChakrabartiStoeckl2021}
Amit Chakrabarti and Manuel Stoeckl.
\newblock The element extraction problem and the cost of determinism and
  limited adaptivity in linear queries.
\newblock Technical report, arXiv:2107.05810, 2021.

\bibitem[CGLMS23]{ChakrabortyGalLaplanteMittalSunny2023}
Sourav Chakraborty, Anna G{\'a}l, Sophie Laplante, Rajat Mittal, and Anupa
  Sunny.
\newblock Certificate games.
\newblock In {\em 14th Innovations in Theoretical Computer Science Conference
  (ITCS 2023)}, volume 251 of {\em Leibniz International Proceedings in
  Informatics (LIPIcs)}, pages 32:1--32:24, 2023.

\bibitem[Faw16]{Fawcett2016}
Joanna~B. Fawcett.
\newblock Locally triangular graphs and normal quotients of the $n$-cube.
\newblock {\em Journal of Algebraic Combinatorics}, 44(1):119--130, 2016.

\bibitem[FF85]{FranklFuredi1985}
Peter Frankl and Zolt{\'a}n F{\"u}redi.
\newblock Forbidding just one intersection.
\newblock {\em Journal of Combinatorial Theory, Series A}, 39(2):160--176,
  1985.

\bibitem[F\"ur83]{Furedi1983}
Zolt{\'a}n F{\"u}redi.
\newblock On finite set-systems whose every intersection is a kernel of a
  star.
\newblock {\em Discrete Mathematics}, 47(1):129--132, 1983.

\bibitem[FSS84]{FurstSaxeSipser1984}
Merrick~L. Furst, James~B. Saxe, and Michael Sipser.
\newblock Parity, circuits, and the polynomial-time hierarchy.
\newblock {\em Mathematical Systems Theory}, 17(1):13--27, 1984.

\bibitem[GH92]{GoldmannHastad1992}
Mikael Goldmann and Johan H{\aa}stad.
\newblock A simple lower bound for the depth of monotone circuits computing
  clique using a communication game.
\newblock {\em Information Processing Letters}, 41(4):221--226, 1992.

\bibitem[GRSS23]{GoosRiazanovSofronovaSokolov2023}
Mika G{\"o}{\"o}s, Artur Riazanov, Anastasia Sofronova, and Dmitry Sokolov.
\newblock Top-down lower bounds for depth-four circuits.
\newblock In {\em 2023 IEEE 64th Annual Symposium on Foundations of Computer
  Science (FOCS)}, pages 1048--1055. IEEE, 2023.

\bibitem[H{\aa}s86]{Hastad1986}
Johan H{\aa}stad.
\newblock Almost optimal lower bounds for small depth circuits.
\newblock In {\em Proceedings of the Eighteenth Annual ACM Symposium on Theory
  of Computing}, pages 6--20. ACM, 1986.

\bibitem[HJP95]{HastadJuknaPudlak1995}
Johan H{\aa}stad, Stasys Jukna, and Pavel Pudl{\'a}k.
\newblock Top-down lower bounds for depth-three circuits.
\newblock {\em Computational Complexity}, 5(2):99--112, 1995.

\bibitem[Hir17]{Hirahara2017}
Shuichi Hirahara.
\newblock A duality between depth-three formulas and approximation by
  depth-two.
\newblock Electronic Colloquium on Computational Complexity, Report
  TR17-092, 2017.

\bibitem[Juk97]{Jukna1997FiniteLimits}
Stasys Jukna.
\newblock Finite limits and monotone computations: The lower bounds criterion.
\newblock In {\em Proceedings of the Twelfth Annual IEEE Conference on
  Computational Complexity}, pages 302--313, 1997.

\bibitem[KW90]{KarchmerWigderson1990}
Mauricio Karchmer and Avi Wigderson.
\newblock Monotone circuits for connectivity require super-logarithmic depth.
\newblock {\em SIAM Journal on Discrete Mathematics}, 3(2):255--265, 1990.

\bibitem[KKM99]{KlavzarKoolenMulder1999}
Sandi Klav\v{z}ar, Jack Koolen, and Henry Martyn Mulder.
\newblock Graphs which locally mirror the hypercube structure.
\newblock {\em Information Processing Letters}, 71(2):87--90, 1999.

\bibitem[LMN93]{LinialMansourNisan1993}
Nathan Linial, Yishay Mansour, and Noam Nisan.
\newblock Constant depth circuits, {Fourier} transform, and learnability.
\newblock {\em Journal of the ACM}, 40(3):607--620, 1993.

\bibitem[MW19]{MW19}
Or~Meir and Avi Wigderson.
\newblock Prediction from partial information and hindsight, with application
  to circuit lower bounds.
\newblock {\em Computational Complexity}, 28(2):145--183, 2019.

\bibitem[Raz87]{Razborov1987}
Alexander~A. Razborov.
\newblock Lower bounds on the size of bounded depth circuits over a complete
  basis with logical addition.
\newblock {\em Mathematical Notes of the Academy of Sciences of the USSR},
  41(4):333--338, 1987.

\bibitem[Ros19]{Rossman2017Subspace}
Benjamin Rossman.
\newblock Subspace-invariant {$\ACzero$} formulas.
\newblock {\em Logical Methods in Computer Science}, 15(3):3:1--3:12, 2019.

\bibitem[Sac63]{Sachs1963}
Horst Sachs.
\newblock Regular graphs with given girth and restricted circuits.
\newblock {\em Journal of the London Mathematical Society},
  s1-38(1):423--429, 1963.

\bibitem[Sip84]{Sipser1984}
Michael Sipser.
\newblock A topological view of some problems in complexity theory.
\newblock In {\em Mathematical Foundations of Computer Science 1984}, volume
  176 of {\em Lecture Notes in Computer Science}, pages 567--572. Springer,
  1984.

\bibitem[Smo87]{Smolensky1987}
Roman Smolensky.
\newblock Algebraic methods in the theory of lower bounds for boolean circuit
  complexity.
\newblock In {\em Proceedings of the Nineteenth Annual ACM Symposium on Theory
  of Computing}, pages 77--82. ACM, 1987.

\bibitem[TZ97]{TardosZwick1997}
G{\'a}bor Tardos and Uri Zwick.
\newblock The communication complexity of the universal relation.
\newblock In {\em Proceedings of the Twelfth Annual IEEE Conference on
  Computational Complexity}, pages 247--259, 1997.

\bibitem[Yao85]{Yao1985}
Andrew Chi-Chih Yao.
\newblock Separating the polynomial-time hierarchy by oracles.
\newblock In {\em Proceedings of the 26th Annual IEEE Symposium on Foundations
  of Computer Science}, pages 1--10, 1985.

\end{thebibliography}
\end{document}